\documentclass{article}
\PassOptionsToPackage{numbers,sort&compress}{natbib}
\usepackage[preprint]{neurips_2026}

\usepackage[utf8]{inputenc}
\usepackage[T1]{fontenc}
\usepackage{amsmath,amssymb,amsthm}
\usepackage{booktabs}
\usepackage{graphicx}
\IfFileExists{microtype.sty}{\usepackage{microtype}}{}
\usepackage{xcolor}
\usepackage{hyperref}
\usepackage{url}
\hypersetup{colorlinks=true,allcolors=[rgb]{0.15,0.25,0.5}}

\newtheorem{proposition}{Proposition}
\newcommand{\kms}{\,\mathrm{km\,s^{-1}}}
\newcommand{\kpc}{\,\mathrm{kpc}}
\newcommand{\Msun}{M_\odot}
\newcommand{\vv}{\mathbf v}
\newcommand{\xx}{\mathbf x}

\title{Closed-Form of the Local Galactic Potential and\\
Stellar Distribution Function from Gaia DR3}

\author{%
  Indranil Das\thanks{Equal contribution.} \\
  Illinois Center for Advanced Studies \\
  of the Universe \& Dept.\ of Physics \\
  University of Illinois Urbana-Champaign \\
  Urbana, IL 61801, USA \\
  \texttt{idas3@illinois.edu} \\
  \And
  Adam Kamoski\footnotemark[1] \\
  Department of Physics \\
  University of Massachusetts Boston \\
  NSF Institute for AI and Fundamental Interactions \\
  \texttt{Adam.Kamoski001@umb.edu} \\
  \AND
  Dora Demiri \\
  European University of Tirana \\
  Tirana, Albania \\
  \texttt{ddemiri@uet.edu.al} \\
  \And
  Brianna Isola \\
  University of New Hampshire \\
  Durham, NH 03824, USA \\
  \texttt{Brianna.Isola@unh.edu} \\
  \AND
  Hanieh Karimi \\
  University of New Hampshire \\
  Durham, NH 03824, USA \\
  \texttt{Hanieh.karimi@unh.edu@unh.edu} \\
  \And
    Dmitrii S.~Zagorulia \\
  Lebedev Physical Institute \\
  Russian Academy of Sciences \\
  Moscow 119991, Russia \\
  \texttt{zagorulia.ds@phystech.edu} \\
}

\makeatletter
\renewcommand{\@noticestring}{}
\makeatother

\begin{document}
\maketitle

\renewcommand{\thefootnote}{\fnsymbol{footnote}}
\renewcommand{\thefootnote}{\arabic{footnote}}

\begin{abstract}
The local dark matter density determines the strength of the signal expected in direct-detection experiments, yet published estimates from stellar motions disagree by more than their errors \citep{read2014,schutz2018,widmark2019,guo2020,salomon2020,widmark2021}, and the most recent machine-learning analysis of Gaia data finds a local density consistent with zero \citep{kalda2025}. According to Jeans' theorem, a distribution function built from integrals of motion satisfies the collisionless Boltzmann equation (CBE) trivially for any choice of potential \citep{jeans1915,binney2008}, so a search that simultaneously fits the distribution function and the potential to the CBE identifies neither. Our pipeline instead estimates the distribution function in isolation, linearizing the equation in terms of accelerations and allowing for direct measurement of the local force field, and then fits closed forms to that field via symbolic regression. Throughout, we find that the usable information lies not in the CBE residual but in the stellar number counts, the observable most distorted by survey selection. Along the vertical profile, our recovered potential agrees with the classical self-gravitating isothermal disc.
\end{abstract}

\section{Introduction}
The Gaia DR3 catalog supplies six-dimensional phase-space coordinates for millions of stars \citep{gaiadr3}, motivating machine-learning recovery of the Galactic potential $\Phi$ and stellar distribution function $f$ from a single snapshot \citep{green2023,kalda2025,an2021}. Symbolic regression \citep{cranmer2020,cranmer2023} provides an interpretable, differentiable output that can be checked against independent measurements. The local dark matter density inferred from such analyses is an input to direct-detection experiments, yet published values span $0.005$--$0.020\,\Msun\mathrm{pc^{-3}}$ with error bars that do not
overlap \citep{read2014,schutz2018,widmark2019,guo2020,salomon2020,widmark2021},
and the most recent neural-network recovery reports a value consistent with zero \citep{kalda2025}. The difficulty is identifiability: the equation relating observation to target is the stationary CBE,
\begin{equation}
\vv\cdot\nabla_{\xx} f - \nabla_{\xx}\Phi\cdot\nabla_{\vv} f = 0,
\label{eq:cbe}
\end{equation}
and when $f$ and $\Phi$ are fit to it jointly, it fails to uniquely identify either; a model can match the data exactly without constraining $\Phi$. Our
pipeline avoids this by fixing $f$ before solving for $\Phi$, applied to $9.9\times10^6$ stars within $1\kpc$.

This work proceeds in three parts: In Sec.~\ref{sec:degen} we provide a quantified analysis of the degeneracy, which consistently puts more than $96\%$ of the constraint on $\Phi$ in the spatial density. In Sec.~\ref{sec:method} we directly measure the acceleration field, with no assumed potential, and we find closed forms for both unknowns under a Poisson-positivity constraint, with a controlled experiment showing the constraint is necessary. In Sec.~\ref{sec:results} we recover a distribution function supervised so that the degenerate solution is no longer the optimum but the uninformative baseline.

%
\section{Degeneracy of the potential function}
\label{sec:degen}
%
\begin{proposition}[Minimizing the residual cannot identify $\Phi$]
\label{prop:degen}
For any static and axisymmetric $\Phi$, $E=\tfrac12|\vv|^2+\Phi$ and $L_z=Rv_\phi$ are integrals of motion. Here, $E$ is the specific orbital energy and $L_z$ is the z-component of angular momentum. Then for any smooth $h$, the choice $f=h(E,L_z)$ satisfies Eq.~\eqref{eq:cbe} exactly. In
particular $f=\mathrm{const}$ satisfies Eq.~\eqref{eq:cbe} for any $\Phi$. The minimum of the residual is therefore zero on a set containing every potential.
\end{proposition}

This follows from Jeans' theorem \citep{jeans1915,binney2008}. A residual-only symbolic search quickly converges to a simple but physically unreasonable zero-loss solution (complexity 2, both components constant). We therefore instead profile $\chi^2$ over the uniform volume density $\rho_{\rm DM}$, which we scan on the interval $\rho_{DM}\in[0,20]\times10^{-3}M\odot\text{pc}^{-3}$ on mocks with known truth (Table~\ref{tab:robustness}). When the distribution function is allowed two free exponential components, the difference between exclusion and inclusion of the spatial density $\nu(z)$ is reflected in a $\chi^2$ span of 0.7 vs. one of 73. For the baseline mock the decomposition is CBE residual $0\%$, conditional velocity shape ${\sim}1\%$, and spatial density ${\sim}99\%$. $\nu(z)$ remains significant under every configuration tested, carrying more than $96\%$ of the constraint in each:
tracer temperature ($12$-$60\kms$), height range ($0.5$-$2\kpc$), and distribution-function family, including King-like truncated models
(Table~\ref{tab:robustness}). This is because for an isothermal tracer $\nu$ depends on $\Phi$ exponentially ($\nu\propto e^{-\Phi/\sigma^2}$) \citep{kuijken1989,read2014,bovyrix2013}, whereas the velocity shape enters only through moments that $f$ can absorb. Since $\nu$ is the observable corrupted by survey completeness, the selection function must be modeled explicitly. When a completeness gradient $S\propto e^{-|z|/\ell}$ is present in a mock and omitted from the model, the recovered surface density shifts by $\sigma^2/2\pi G\ell$, at a loss floor indistinguishable from the clean run with $S=1$
(App.~\ref{app:mocks}).

\begin{proposition}[Two-integral models are meridionally isotropic]
\label{prop:aniso}
For static axisymmetric $\Phi$, every $f(E,L_z)$ satisfies $\sigma_R=\sigma_z$ exactly at every point.
\end{proposition}
\vspace{-0.5em}
\begin{proof}
$E$ depends on $v_R$ and $v_z$ only through $v_R^2+v_z^2$, and $L_z$ not at all, so
both are invariant under the exchange $v_R\!\leftrightarrow\!v_z$. Therefore, so is $f$,
and $\langle v_R^2\rangle=\langle v_z^2\rangle$ at fixed $(R,z)$. 
\end{proof}

The two-integral family that minimizes the residual is the same family constrained to $\sigma_R=\sigma_z$. Hence, a measured anisotropy $\sigma_R\neq\sigma_z$ distinguishes our solution from it. This does not distinguish it from every stationary solution, since a three-integral $f(E,L_z,I_3)$ satisfies Eq.~\eqref{eq:cbe} identically and is generically anisotropic. An exactly separable third integral requires $\Phi$ of St\"ackel form, which we neither impose nor test for (App.~\ref{app:abl}). Therefore, we expect a non-zero residual.

%
\section{Method}
\label{sec:method}
\textbf{Selection.} The $1\kpc$ sphere holds $9.94\times10^6$ stars with \texttt{RUWE}$<1.4$, parallax S/N $>10$, $\ge4$ radial-velocity transits, and velocity uncertainty $<10\kms$, split 70/15/15 before any fitting (in-sample
selection picks the wrong front member on mocks; App.~\ref{app:abl}). Gradients of $\ln f$ are evaluated on a $2.5\times10^6$-star subsample; $1.73\times10^6$ of those fall in the 94 acceleration cells, and the symbolic search for $\ln f$ is fitted on $2.0\times10^4$ rows drawn from the $2.21\times10^6$ stars within $0.9\kpc$. We use $R_0=8.122\kpc$ and $z_\odot=0$, the latter by construction of the catalogue frame (App.~\ref{app:data}). The observed number density per unit volume falls to $5.1\%$ of its value at $0.15\kpc$ across the sphere, with $|\mathrm d\ln S/\mathrm dd|$ peaking at $10.9\kpc^{-1}$ against a physical $|\partial_z\ln\nu|$ of $3.0\kpc^{-1}$ (Fig.~\ref{fig:selection}). We fit $\mu=V\!\cdot\!S(d)A(\ell,b)\nu(R,z)$ as a Poisson GLM. The three factors are identifiable up to two overall multiplicative constants because they are different functions of the same point: stars at equal distance in different directions lie at different heights. It returns a $2.88\kpc$ radial scale length
and a two-scale-height vertical profile whose local scale height at $300$~pc is $334$~pc, neither constrained to a literature value. The complete pipeline and code are available at \url{https://github.com/briannaisola/GaiaSR}.

\textbf{Empirical $f$.} We factorize $f=\nu\,p(\vv|\xx)$ and fit the conditional with a 96-component full-covariance Gaussian mixture, whose log-gradient is closed-form and agrees with central finite differences to $10^{-6}$ in all six coordinates. The conditional is selection-free:
$S(\xx)f/\!\int\!S(\xx)f\,\mathrm d^3v=f/\!\int\!f\,\mathrm d^3v$, since $S$ has no velocity dependence, so components are selected on held-out conditional log-likelihood. Refitting on a disjoint set of $5\times10^5$ stars moves
$\partial_R\ln f$ by $0.815$ against its own root mean square (rms) of $1.673$; refitting with Gaia's
errors injected a second time moves it by $0.736$. The two shifts are comparable
in size, so the limiting factor is the density estimator, and the implied residual floor is $48.7\kms\kpc^{-1}$ for any expression carrying its gradients. The azimuthal streaming term carries $93\%$ of the total rms ($222.2$ against
$67.4$ and $41.7\kms\kpc^{-1}$), is measured over only a $14^\circ$ baseline in
azimuth, and cannot be balanced by any axisymmetric force; dropping it moves
$\mathrm dv_c/\mathrm dR$ from $-9.0$ to $-1.1$ against a literature
$-1.7\pm0.1$ \citep{eilers2019}.

\begin{figure}[t]\centering
\includegraphics[width=\textwidth]{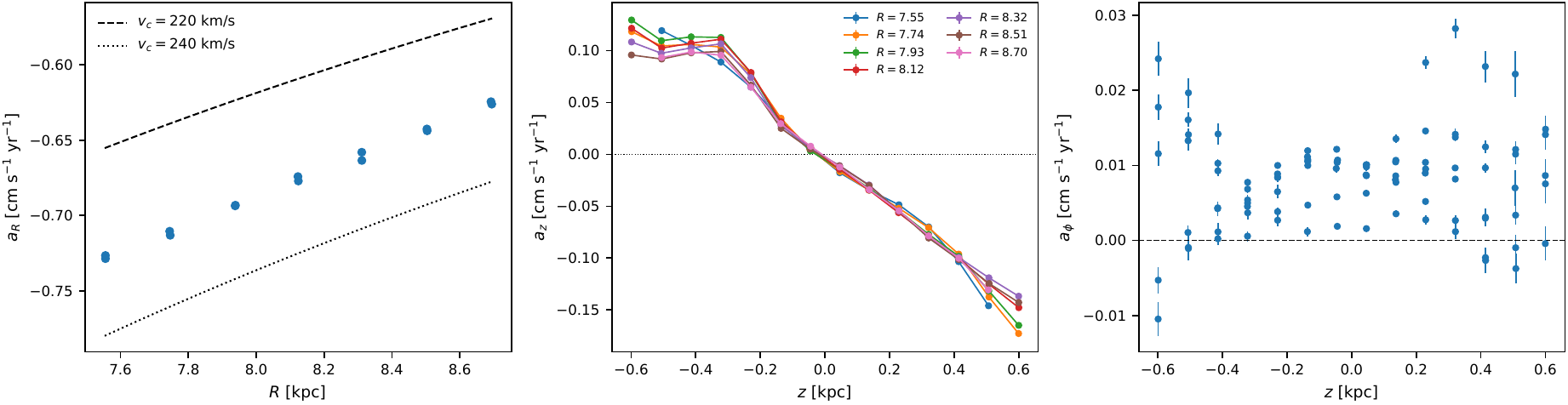}
\caption{The acceleration field $a$, measured cell by cell with no functional form
assumed for $\Phi$. Each point is one spatial cell solved from thousands of linear
equations, Eq.~\eqref{eq:split}; error bars on the rightmost plot are the cell-fit standard errors.
Left: the radial component $a_r$ at the midplane, with constant-$v_c$ curves for scale.
Center: the vertical component $a_z$, colored by radius, changing sign at the dynamical
midplane. Right: the azimuthal component $a_\phi$, which an axisymmetric $\Phi$ requires to vanish and which instead comes out systematically positive at the per-cent level of $a_R$.}
\label{fig:accel}
\vspace{-0.6em}
\end{figure}

%
\textbf{Acceleration Measurement.} Eq.~\eqref{eq:cbe} is linear in $\mathbf
a$. In cylindrical coordinates with $g=\ln f$, grouped by whether a term
involves the potential,
\begin{equation}
\underbrace{v_R\partial_R g+v_z\partial_z g+\tfrac{v_\phi^2}{R}\partial_{v_R}g
-\tfrac{v_Rv_\phi}{R}\partial_{v_\phi}g}_{\textstyle K_i}\;+\;\mathbf
W_i\!\cdot\!\mathbf a\;=\;0,
\label{eq:split}
\end{equation}
with $\mathbf W_i=\nabla_{\vv}\,g$ and the azimuthal streaming term dropped as
above. The centrifugal and Coriolis terms arise from the rotation of the cylindrical
basis as a star moves, and are fixed by the star's own coordinates. 
We solve $94$ acceleration cells ($7$ radial bins $\times$ $14$ vertical bins $-4$ star-count-cut cells) by Huber-reweighted least squares, assuming no functional form for $\Phi$ (Fig.~\ref{fig:accel}). The field comes out at $\sim0.6\,\mathrm{cm\,s^{-1}\,yr^{-1}}$, the acceleration
undetectable in any single star over the mission lifetime but recoverable from the collective statistics of many. From it follow $v_c=\sqrt{-a_RR}$,
$\Sigma(<|z|)=|a_z|/2\pi G$, and the dynamical midplane where $a_z$ changes sign. The whole volume as one cell gives $v_c=231.4\kms$ against $229\pm3$ \citep{eilers2019}; the cells give $\Sigma(<0.5\kpc)=44.0\,\Msun\mathrm{pc^{-2}}$, between the $41$ and $65\pm6$ measured at $0.35$ and $0.8\kpc$ \citep{holmberg2004}, and place the midplane $18$~pc from
the Sun against a photometric $20.8\pm0.3$~pc \citep{bennett2019}. 

The quoted errors are least-squares statistical errors on millions of stars, and they understate the real uncertainty: nothing below about a per cent here is resolved. Solving for $a_\phi$ rather than imposing $a_\phi=0$ bounds the departure from axisymmetry at $1.3\%$ of $|a_R|$, comparable to that floor, and we return to it in the limitations.

%
\textbf{Symbolic $\Phi$.} We search with PySR \citep{cranmer2023,pysrtemplates}.
Each cell enters the training set twice, tagged by an indicator, so a single
$\Phi$ must reproduce both force components and the fitted field is curl-free by
construction. $5.24\ln R\pm8z^2$ ($\nabla^2\Phi=\pm16$) generate equally valid acceleration fields $a$ that the data cannot tell apart, and only the sign of the Laplacian distinguishes them. Hence, a Poisson positivity constraint becomes necessary, entering the objective as
$\mathcal L_\Phi=\sigma^{-2}(\text{pred}-a_{\rm tgt})^2+
\lambda[\min(\nabla^2\Phi,0)]^2$ with $\lambda=25$, one-sided so that it vanishes on the admissible set and cannot bias the choice among physical candidates. From the Pareto front we take the cheapest expression within a factor of two of the best admissible loss, subject to (1) $\rho>-0.005$ throughout the $(R,z)$ domain where Poisson positivity is enforced and (2) a plausibility window $\rho(R_0,0)\in[0,0.5]\,\Msun\mathrm{pc^{-3}}$, $v_c\in[150,320]\kms$. This
window is wide enough to exclude only nonsense: all 21 of the 26 front rows surviving Poisson positivity already lie inside it, and removing it returns the
same expression. Selection is therefore by positivity and parsimony, and the window did not manufacture the $v_c$ and $\rho(R_0,0)$ reported below.

%
\textbf{Closed-form $f$.} By Prop.~\ref{prop:degen}, the residual cannot serve as the objective. We instead
regress $g$ onto the empirical $\ln\hat f$ and its five gradients. The loss is
$\mathcal L_f=\langle((g-\ln\hat f)/S_g)^2\rangle+\tfrac15\sum_i\langle((\partial_i g
-\partial_i\ln\hat f)/S_i)^2\rangle$, averaged over stars, where each $S$ is the
rms of the target it normalizes. The gradients are supervised because Eq.~\eqref{eq:split} absorbs them, and because an expression can track a function's values closely while its derivatives wander. A constant $g$ scores exactly $2.000$, one from the value term and one from the five gradient terms, whose model gradients vanish. An expression carrying no information about the target can get a score as low as $2.000$ under this loss. The same expression scores $0$ under a residual objective, where it wins. Selection is on $R^2$ against $\ln\hat f$, taken over Pareto front rows that are finite on more than $99\%$ of held-out stars. The residual and the anisotropy are computed afterwards
and enter no decision. Two of the five supervised gradients are $\partial_{v_R}\ln\hat f$ and $\partial_{v_z}\ln\hat f$, which carry $\sigma_R$
and $\sigma_z$ for a locally Gaussian conditional. The anisotropy reported below
is therefore derived from supervised quantities. Front members of comparable loss return ratios between
$0.91$ and $9.90$ (Table~\ref{tab:front}), so reproducing the supervised
gradients pointwise does not by itself reproduce the ellipsoid they imply.

\section{Results}
\label{sec:results}
\begin{figure}[t]\centering
\includegraphics[width=0.94\textwidth]{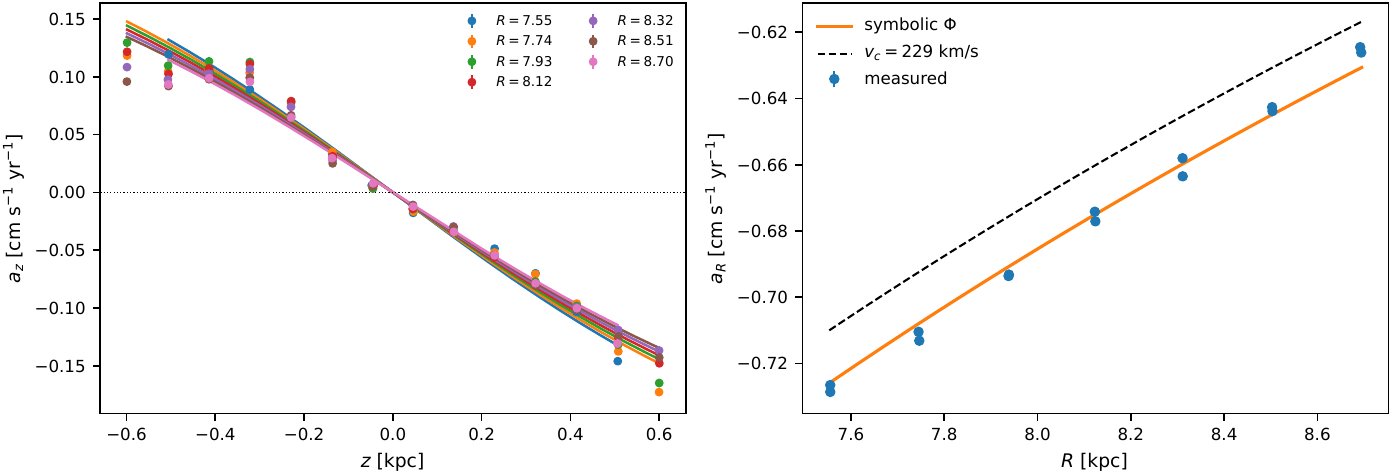}
\caption{Two plots representing the closed-form potential $\Phi(R,z)$ of Eq.~\eqref{eq:phi}, containing eleven nodes and two
fitted constants, against the independently measured field of
Fig.~\ref{fig:accel}. Left: vertical acceleration $a_z$, points measured and lines the
symbolic $\Phi$, colored by radius $R$. Right: radial acceleration $a_R$ at the midplane,
with a constant-$v_c$ curve at the literature value for comparison. The fit is to
the measured field alone.}
\label{fig:money}
\vspace{-0.6em}
\end{figure}
The potential comes out at complexity 11 as
\begin{equation}
\Phi(R,z)=5.3618\,\ln\!\big(R+0.39193\,\ln\cosh z\big),
\label{eq:phi}
\end{equation}
with $R,z$ in kpc and $\Phi$ in $(100\kms)^2$.
The recovered expression contains a \(\ln\cosh z\) vertical dependence characteristic of the potential of a self-gravitating isothermal sheet, whose density profile \(\nu\propto\mathrm{sech}^2(z/2h)\) was derived by \citet{spitzer1942}. Notably, this functional form emerged from the free operator set rather than being imposed a priori.

The distribution function comes out at complexity 44 as
\begin{align}
\ln f = 2.407
&-\ln\cosh(8.523\,v_z)
-\ln\cosh\!\big(z\ln\cosh(v_\phi\ln R)\big)\nonumber\\
&-\ln\cosh\!\big[(\ln\cosh v_\phi-1.625)(\ln\cosh(5.251\,v_z)-R+1.385)\big]\nonumber\\
&-\ln\cosh\!\big(\ln\cosh(2.109\,v_Rv_\phi)\big).
\label{eq:lnf}
\end{align}
The selected expression is non-separable in the velocity components: $v_R$ enters through the product $v_Rv_\phi$, while the large-$|v_z|$ behavior is approximately exponential.

\begin{table}[t]
\begin{minipage}[t]{0.49\textwidth}\centering\small
\setlength{\tabcolsep}{3pt}
\caption{Pareto front for $\ln f$, held out; every fifth member plus the two
endpoints, from 37 (full front in App.~\ref{app:abl}). Row two is
$g=\mathrm{const}$, the exact global optimum of a residual objective, present on
our own front and rejected on $R^2$. $\epsilon_\nabla$ is the mean relative
gradient error.}
\label{tab:front}
\vspace{-0.2em}
\begin{tabular}{rrrrrr}
\toprule
$c$ & loss & $R^2$ & $\epsilon_\nabla$ & CBE & $\sigma_R/\sigma_z$\\
\midrule
1 & 2.016 & $-8\!\times\!10^{-3}$ & 1.00 & 6.2 & 1.14\\
2 & 2.000 & $-3\!\times\!10^{-5}$ & 1.00 & \textbf{0.00} & 0.91\\
9 & 1.352 & 0.510 & 0.90 & 41.7 & 8.11\\
10 & 1.313 & $-2\!\times\!10^{9}$ & $3\!\times\!10^{4}$ & $1\!\times\!10^{8}$ & 9.90\\
20 & 0.879 & 0.745 & 0.76 & 51.9 & 1.37\\
28 & 0.691 & 0.876 & 0.74 & 41.7 & 1.87\\
39 & 0.603 & 0.917 & 0.70 & 36.1 & 2.13\\
\textbf{44} & \textbf{0.589} & \textbf{0.923} & \textbf{0.69} & \textbf{38.2} & \textbf{1.88}\\
\bottomrule
\end{tabular}
\end{minipage}\hfill
\begin{minipage}[t]{0.49\textwidth}\centering\small
\setlength{\tabcolsep}{3pt}
\caption{Held-out results. $\sigma_R/\sigma_z$ is measured from the same sample,
not literature; the rest are independent
\citep{eilers2019,bennett2019,mckee2015,read2014, holmberg2004}. Two of the five external
comparisons fail, marked $\dagger$. Uncertainties are discussed in the text and
are not the statistical errors. $\ast$: measured from the same sample.
$\ddagger$: the supervising mixture on the same axisymmetric target. CBE residual
in $\kms\kpc^{-1}$.}
\label{tab:results}
\vspace{-0.2em}
\begin{tabular}{lrr}
\toprule
Quantity & This work & Reference\\
\midrule
$\sigma_R/\sigma_z$ & $1.83$--$2.13$ & $1.893^\ast$\\
$v_c(R_0)$ & $231.6\kms$ & $229\pm3$\\
$\Sigma(<0.5\kpc)$ & $44.0$ & 41--65 (0.35-0.8 kpc)\\
midplane offset & 18 pc & $20.8\pm0.3$\\
$\mathrm dv_c/\mathrm dR^\dagger$ & $-1.1$ & $-1.7\pm0.1$\\
$\rho(R_0,0)^\dagger$ & $0.048$ & 0.084--0.10\\
\midrule
CBE resid. & 38.2 & $58.6^\ddagger$\\
rel.\ to $\Phi{=}0$ & 0.204 & 1.000\\
$R^2(\ln f)$ & 0.923 & ---\\
\bottomrule
\end{tabular}
\end{minipage}
\vspace{-0.8em}
\end{table}

The recovered field reproduces three of the five external comparisons in Table~\ref{tab:results}. Table~\ref{tab:front} shows that the complexity-2 constant solution achieves zero CBE residual but no meaningful fit to the empirical distribution ($R^2=-3\times10^{-5}$), demonstrating the failure of residual-only model selection.
$R^2$
then rises with complexity, though not monotonically along the whole front,
and the meridional anisotropy settles into the range $1.83$--$2.13$ over the top
twelve members, mean $1.94$, bracketing the measured $1.893$. The
maximum-$R^2$ member reads $1.88$, but we quote the spread, since no member
reaches the observed ratio to better than the $5\%$ scatter of the front. The recovered $\Phi$ and $\ln f$ predict this ellipsoid because two of the
supervised gradients carry it (Sec.~\ref{sec:method}). That they come out at the right value shows the fit preserved the velocity shape it was given.

$\Phi=0$ gives exactly $1.000$, so $0.204$
means the potential terms cancel most of the variance of the remaining,
axisymmetric part of Eq.~\eqref{eq:split}. This residual is evaluated with the
azimuthal streaming term $\tfrac{v_\phi}{R}\partial_\phi g$ removed, which no
axisymmetric force can balance and which alone carries $93\%$ of the signal
(Sec.~\ref{sec:method}). With it retained the relative residual is $0.76$. The
value $38.2$ lies below the $58.6$ scored on the same axisymmetric target by the
mixture that supervised it, because a smooth expression cannot reproduce estimator
noise. The $\Phi$ residual is constant to $0.5\%$ across the train, validation,
and test splits, from two fitted constants against $1.7\times10^6$ evaluation
stars, though the symbolic search itself saw only $2\times10^4$.

\textbf{Ablations} (App.~\ref{app:abl}). Separable potentials
$\Phi=\Phi_R(R)+\Phi_z(z)$ occupy complexity 9--10 and are $43\%$ worse than the
coupled form, which predicts $a_z$ falling $22\%$ across the sample where a
separable form predicts no change; since $E_z$ is conserved only if
$\partial^2\Phi/\partial R\,\partial z=0$, the recovered coupling is also why
$E_z$ is not available as a third integral here. Among vertical profiles fitted
freely, $\ln\cosh$ ($\chi^2/\mathrm{dof}=275.4$) is matched by
$\sqrt{z^2+h^2}-h$ ($275.8$) and beaten by nothing, while any form with a midplane
kink is $10\times$ worse. Rescaling inputs by measured scales moves the residual to
$23.69$ and $R^2$ to $0.937$ at identical complexity. The observed anisotropy is
stable to $\pm1\%$ across apertures from $|z|<0.05$ to $|z|<0.30\kpc$.

\textbf{Where it departs.} The local density is $\rho=0.048$ against a literature
$0.084$--$0.10\,\Msun\mathrm{pc^{-3}}$ \citep{mckee2015,read2014}. Evaluating $\rho=-(4\pi G)^{-1}[\partial_Ra_R+a_R/R+\partial_za_z]$ on the measured field, with no symbolic expression involved, gives $0.0493$: the fit reproduces the field to $3\%$, so the deficit is already present in the accelerations and enters
upstream through the density estimator, exactly where the mocks place it, a
smoothed $\nu$ moving the recovered split from its clean value $(\Sigma,\rho)=(50,0.010)$ to
$(66,0.003)$ because $\Sigma$ integrates the vertical force while $\rho$
differentiates it (App.~\ref{app:mocks}).

\section{Conclusion}

We present closed-form models of the local Galactic potential and stellar
distribution function inferred from the Gaia DR3 catalog. Controlled mock tests show
that stellar number counts provide the main constraint on the potential,
whereas jointly fitting $f$ and $\Phi$ to the stationary collisionless
Boltzmann equation alone remains degenerate. From the completeness-corrected
distribution function, we reconstruct the acceleration field to which
the symbolic potential is fitted under a Poisson-positivity constraint.
The recovered potential contains a $\ln\cosh$ term characteristic of a
self-gravitating isothermal sheet. Three of five external benchmarks are
broadly reproduced, and the top Pareto-front models bracket the observed
velocity anisotropy. Steady state and axisymmetry are assumed within $1\kpc$, despite a
non-zero measured azimuthal acceleration. The $\ln\cosh$ term has a fixed
vertical scale of $2h=1\kpc$. Although the surface density is approximately
recovered, the local density is underestimated by about a factor of two,
a deficit already present in the empirical acceleration field.
The selection model has an $11\%$ systematic in $\partial_z\ln\nu$
that dominates the uncertainty in $\Phi$.

\section*{Acknowledgements}
We thank Akshay Ghalsasi for discussions, and Miles Cranmer and Jose M.~Munoz for
posing the problem. We thank IAIFI for inspiration and for facilitating the work through computing resources provided during during and after the Hackathon where the problem was posed. We thank NCSA and ACCESS/PSC for computing resources.

This work has made use of data from the European Space Agency (ESA) mission
\emph{Gaia} (\url{https://www.cosmos.esa.int/gaia}), processed by the \emph{Gaia}
Data Processing and Analysis Consortium (DPAC,
\url{https://www.cosmos.esa.int/web/gaia/dpac/consortium}). Funding for the DPAC
has been provided by national institutions, in particular the institutions
participating in the \emph{Gaia} Multilateral Agreement.

{\small
\setlength{\bibsep}{0pt plus 0.3ex}
\bibliographystyle{unsrtnat}
\bibliography{biblio}}

\clearpage
\appendix
\setcounter{table}{0}\renewcommand{\thetable}{A\arabic{table}}
\setcounter{equation}{0}\renewcommand{\theequation}{A\arabic{equation}}
\setcounter{figure}{0}\renewcommand{\thefigure}{A\arabic{figure}}

\section{Mock anatomy of the degeneracy}
\label{app:mocks}

\textbf{Setup.} Our mock inputs are
$\Phi(z)=2\pi G\Sigma(2h)\ln\cosh(z/2h)+2\pi G\rho_{\rm DM}z^2$, $\Sigma=48\,\Msun\mathrm{pc^{-2}}$, $h=0.20\kpc$, and $\rho_{\rm DM}=0.010\,\Msun\mathrm{pc^{-3}}$. When run clean without the injected completeness gradient, we recover $(\Sigma,\rho_{\rm DM})=(50,0.010)$ from this mock. The true distribution function is
$F(E)=0.7\,e^{-E/\sigma_1^2}+0.3\,e^{-E/\sigma_2^2}$, with
$(\sigma_1,\sigma_2)=(18,40)\kms$ and $E=v^2/2+\Phi$.

\textbf{Profiling procedure.} For each trial $\rho_{\rm DM}$ we fix $\Phi$ and
optimize $F$ (as $\sum_i a_ie^{-E/\sigma_i^2}$, $a_i\geq0$, $\sum a_i=1$) to
minimize $\chi^2$ against the conditional dispersion and kurtosis
$(\sigma_z,\kappa)$ in 10 height bins over $0.05$--$0.95\kpc$, with assigned
errors of $2\%$ and $3\%$; when the spatial density is included the profile
scores additionally against $\nu(z)$ normalized at $z=0$ with $2\%$ errors.
Eight random restarts ensure convergence. These $\Delta\chi^2$ spans are computed
at fixed assumed errors and are heuristic; the percentages depend on the relative
error scaling, though the ordering $\nu\gg P(v|z)$ is insensitive to it over the
tested $1$--$5\%$ range. Table~\ref{tab:robustness} gives the
full sweep behind the decomposition of Sec.~\ref{sec:degen}.

\begin{table}[h]\centering\small
\caption{The sensitivity hierarchy across mock configurations: $\Delta\chi^2$
span over $\rho_{\rm DM}$ from the conditional velocity shape alone against the
shape plus spatial density, with the distribution function free (three
exponential components unless noted). Spans are calculated over the interval $\rho_{DM}\in[0,20]\times10^{-3}M\odot\text{pc}^{-3}$. The $\nu(z)$ share is
$1-\text{span}_v/\text{span}_{+\nu}$. Each block varies one axis about the
baseline and the two sweeps were run independently, so the baseline row differs
between them. The $\nu(z)$ share stays above $96\%$ throughout and above $98\%$
for every aperture $|z|_{\max}\geq0.75\kpc$.}
\label{tab:robustness}
\begin{tabular}{lrrr}
\toprule
Configuration & Span $P(v|z)$ & Span $+\,\nu(z)$ & $\nu(z)$ share\\
\midrule
\multicolumn{4}{l}{\emph{Tracer temperature}}\\
$(\sigma_1,\sigma_2)=(12,30)\kms$ & 0.3 & 91 & 99.7\%\\
$(\sigma_1,\sigma_2)=(18,40)$ [baseline] & 0.6 & 57 & 98.9\%\\
$(\sigma_1,\sigma_2)=(25,50)$ & 0.5 & 29 & 98.3\%\\
$(\sigma_1,\sigma_2)=(35,60)$ & 0.2 & 16 & 98.8\%\\
isothermal ($\sigma=15$) & 0.0 & 281 & 100\%\\
\midrule
\multicolumn{4}{l}{\emph{Height range}}\\
$|z|_{\max}=0.5\kpc$ & 0.1 & 3.0 & 96.7\%\\
$|z|_{\max}=0.75\kpc$ & 0.3 & 20 & 98.5\%\\
$|z|_{\max}=1.0\kpc$ [baseline] & 0.7 & 73 & 99.0\%\\
$|z|_{\max}=1.5\kpc$ & 0.9 & 209 & 99.6\%\\
$|z|_{\max}=2.0\kpc$ & 0.7 & 438 & 99.8\%\\
\midrule
\multicolumn{4}{l}{\emph{Distribution-function family}}\\
$F$ with 2 free components & 0.7 & 73 & 99.0\%\\
$F$ with 5 free components & 0.4 & 73 & 99.5\%\\
King-like (truncated at $80\kms$) & 0.0 & 303 & 100\%\\
\bottomrule
\end{tabular}
\end{table}

\begin{figure}[h]\centering
\includegraphics[width=0.9\textwidth]{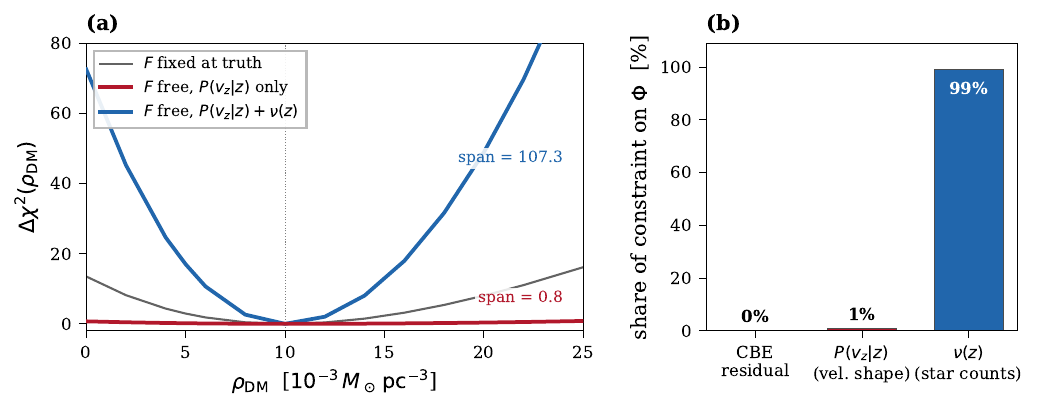}
\caption{Spans are calculated over the interval $\rho_{DM}\in[0,25]\times10^{-3}M\odot\text{pc}^{-3}$, accounting for their slight deviation from corresponding spans reported in Table~\ref{tab:robustness}. Freeing the distribution function destroys the constraint on $\rho_{\rm DM}$; adding the star counts restores it. Left: $\Delta\chi^2$ over
$\rho_{\rm DM}$ for the baseline mock, with the distribution function fixed at
truth (gray), free and scored on the conditional velocity shape alone (red), and free with the spatial density added (blue).
Right: the same three profiles as fractions of the total.}
\label{fig:decomp}
\end{figure}

\textbf{An injected completeness gradient.} Drawing the baseline mock through
$S\propto e^{-|z|/\ell}$ with $\ell=0.6\kpc$ and fitting with $S$ ignored, the
recovered potential converges on $\Phi_{\rm eff}=\Phi+\sigma^2|z|/\ell$: an
isothermal tracer cannot distinguish an unmodeled completeness gradient from a
thin sheet of surface density
$\Delta\Sigma=\sigma^2/2\pi G\ell=25\,\Msun\mathrm{pc^{-2}}$. The recovered
surface density reads $74$ against the predicted $75$, the clean recovery
plus $\Delta\Sigma$; the scale height
collapses to ${\sim}120$~pc; and $\rho_{\rm DM}$ scatters over $0.003$--$0.027$
with complexity and fitting range where the clean recovery is stable at
$0.010$. The loss floor is indistinguishable from the clean run's, so no goodness-of-fit or Pareto criterion detects the bias. This is why Sec.~\ref{sec:method} models $S$ explicitly.

\textbf{A smoothed density estimate.} Estimating $\nu$ with a Gaussian-smoothed
histogram (1.2 cells) deforms $\ln f$ by only $4\%$ at the midplane, where its
curvature peaks, yet pulls the recovered disc--halo split from
$(\Sigma,\rho)=(50,0.010)$ to $(66,0.003)$; the smoothed surface is fit better than the truth fits it. $\Sigma$ integrates the vertical force
while $\rho$ differentiates it, so smoothing $\nu$ moves density out of the halo
term and into the sheet. That is the $\rho(R_0,0)$ deficit of
Sec.~\ref{sec:results}, reproduced here where the truth is known.

\section{Data and the selection model}
\label{app:data}

\begin{figure}[h]\centering
\includegraphics[width=0.9\textwidth]{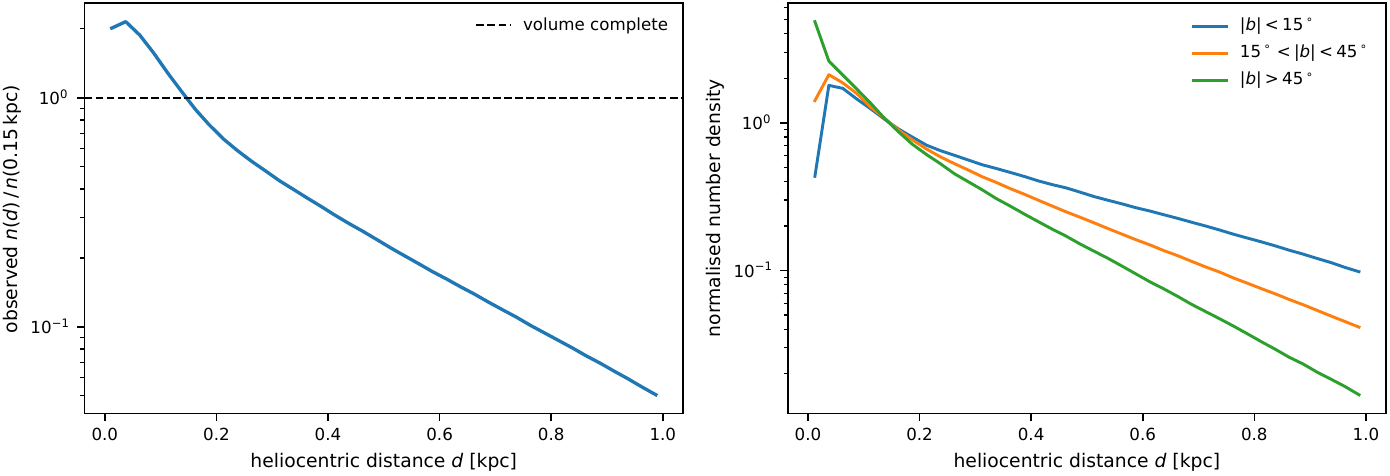}
\caption{Survey completeness dominates the observed density. Left: observed
number density per unit volume against heliocentric distance, normalized at
$0.15\kpc$. It falls by more than an order of magnitude across the ball, driven by a selection gradient $\lvert\mathrm d\ln S/\mathrm dd\rvert$ of $10.9\kpc^{-1}$ against a physical $\lvert\partial_z\ln\nu\rvert$ of $3.0\kpc^{-1}$. Right: the same
profile split by galactic latitude. The three lines separate because the
distance falloff mixes the isotropic selection $S(d)$ with the height-dependent
stratification $\nu(R,z)$; the GLM separates them because stars at equal
distance in different directions lie at different heights.}
\label{fig:selection}
\end{figure}

We adopt $R_0=8.122\kpc$ and $z_\odot=0$: the catalogue is built in a frame
centred on the Sun, so the dynamical midplane offset of Sec.~\ref{sec:method} is
measured against a fixed $z=0$. The GLM
$\mu=V\!\cdot\!S(d)\,A(\ell,b)\,\nu(R,z)$ fits the three factors jointly. The
density factor returns a $334$~pc tracer scale height and a $2.88\kpc$ scale
length, neither tied to literature values. Dust is not separable in distance and direction, so we refit on the dust-poor $|b|>20^\circ$ sight lines: $\partial_z\ln\nu(0.3\kpc)$ moves from $-2.99$ to $-2.67\kpc^{-1}$, an $11\%$ systematic carried forward, and by the hierarchy of Sec.~\ref{sec:degen} that
systematic is what limits $\Phi$.

\begin{figure}[h]\centering
\includegraphics[width=0.85\textwidth]{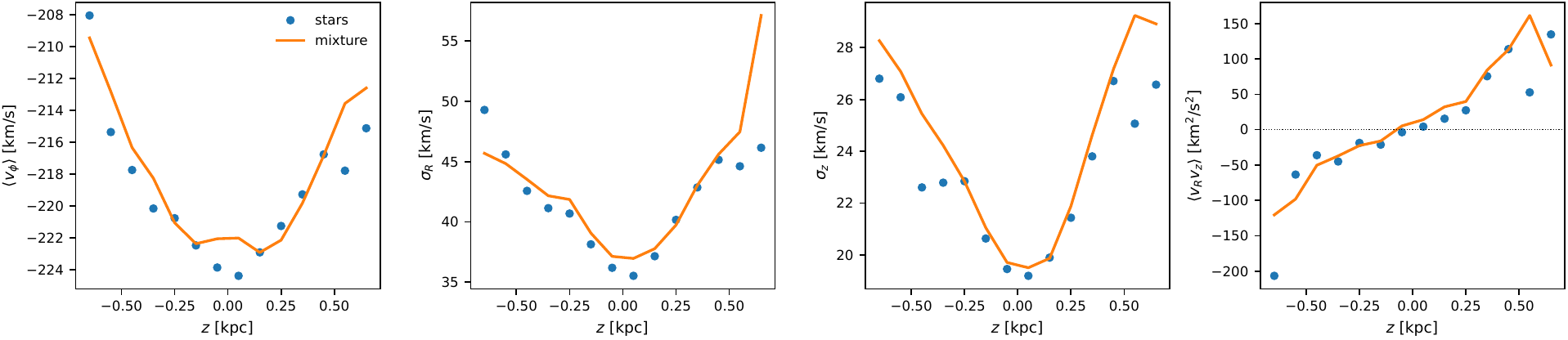}
\caption{The 96-component mixture against conditional velocity moments it was
never shown: held-out stars (points) versus the mixture (line), as functions of
height. The mixture is fitted to individual stars, so the velocity moments serve as an independent check. The mixture reproduces them near the midplane and drifts from
the dispersions beyond $|z|\sim0.4\kpc$, one source of the estimator floor
discussed in Sec.~\ref{sec:method}.}
\label{fig:gmmcheck}
\end{figure}

\section{Ablations and model-selection protocol}
\label{app:abl}

\textbf{Separable potentials.} Constraining $\Phi=\Phi_R(R)+\Phi_z(z)$ costs
$43\%$ in loss at complexity 9--10 against the coupled form of
Eq.~\eqref{eq:phi}. The coupling accounts for the $22\%$ decline of $a_z$ across
the sampled radial range, which a separable form cannot produce. The same cross term removes $E_z$ as a candidate third integral, since $E_z$ is conserved only when $\partial^2\Phi/\partial R\partial z=0$; this excludes the separable case but does not exclude the St\"ackel family as a whole, so a non-zero CBE residual is expected here (Sec.~\ref{sec:degen}).

\textbf{Vertical profile family.} Fitting the vertical term freely among
candidate profiles: $\ln\cosh$ reaches $\chi^2/\mathrm{dof}=275.4$, the
softened modulus $\sqrt{z^2+h^2}-h$ ties at $275.8$, and every form with a
midplane kink (e.g. $|z|$) is an order of magnitude worse. The data favor the
isothermal-sheet family but do not distinguish its two smooth
parameterizations, which agree to the width of the fitted region.

\textbf{Input scaling.} Rescaling the inputs by the measured scale height and
scale length moves the residual from $38.15$ to $23.69\kms\kpc^{-1}$ and
$R^2$ from $0.923$ to $0.937$ at identical complexity: unit choices are worth
as much as ${\sim}10$ complexity points, so all searches run in scaled
variables.

\textbf{Aperture stability.} The measured anisotropy $\sigma_R/\sigma_z$ is
stable to $\pm1\%$ across apertures from $|z|<0.05$ to $|z|<0.30\kpc$, so its
value is not an artifact of the aperture; the spread we report in
Sec.~\ref{sec:results} comes from the choice of front member.

\textbf{Why selection is held out.} The 70/15/15 split of
Sec.~\ref{sec:method} is fixed before fitting because in-sample selection picks
the wrong front member. On mocks with known truth, in-sample $\chi^2$ prefers a
spuriously curved member over the straight one that is true. Held-out $\chi^2$
reverses the choice at every complexity (in-sample $\chi^2/N$ of $1.21$ against
held-out $1.53$ for the line, $0.89$ against $1.71$ for the exponential). An
exponential mimics a line over a bounded range, and only held-out scoring
resists it. Every selection in this work follows the rule: mixture components on
held-out conditional likelihood, $\ln f$ on held-out $R^2$, $\Phi$ on the
screened front.

\textbf{Compute.} The acceleration solve is linear least squares over 94 cells
and runs in minutes on a single CPU node. Each PySR search (potential and
distribution function) runs on one multi-core CPU node in a few hours; the
full set of ablations reported here is ${\sim}10^3$ CPU-core-hours.

\end{document}